\documentclass[letterpaper]{article} 
\usepackage{aaai20}  
\usepackage{times}  
\usepackage{helvet} 
\usepackage{courier}  
\usepackage[hyphens]{url}  
\usepackage{graphicx} 
\usepackage{graphicx}  
\nocopyright
\usepackage{enumitem}
\setlist[enumerate]{leftmargin=*}
\usepackage{verbatim}
\usepackage{amsfonts}
\usepackage{tikz}
\usepackage{amssymb}
\usepackage{amsthm}
\usepackage{float}
\usetikzlibrary{fit}
\usepackage{amsmath}

\usepackage{subcaption}
\newtheorem{definition}{Definition}

\newtheorem{theorem}{Theorem}
\newtheorem{lemma}{Lemma}

\newtheorem{proposition}{Proposition}
\newtheorem{example}{Example}

\usepackage{microtype}
\usepackage{mathtools}
\usepackage{amssymb,textcomp}
\usepackage{latexsym}
\usepackage{amsmath,amsfonts}
\usepackage{enumitem}
\usepackage{graphicx}
\usepackage{amsthm}
\usepackage{xcolor}
\usepackage{xspace}
\usepackage{amstext}

\theoremstyle{definition}

\makeatletter
\DeclareRobustCommand*\cal{\@fontswitch\relax\mathcal}
\makeatother

\newcommand{\submissionversion}[1]{}

\usepackage{xargs}
\usetikzlibrary{arrows.meta}
\newcommandx*{\triplearrow}[4][1=0, 2=2.5]{
	\draw[line width=1pt,double distance=5,#3] #4;
	\draw[line width=1pt,double distance=1pt,shorten <=#1,shorten >=#2,#3] #4;
}
\usetikzlibrary{shapes}
\usepackage{pgfplots}
\usepackage{pgffor}

\title{Convergence of Opinion Diffusion is PSPACE-complete}
\author{\Large \textbf{Dmitry Chistikov\textsuperscript{\rm 1} \and Grzegorz Lisowski\textsuperscript{\rm 1} \and Mike Paterson\textsuperscript{\rm 1} \and Paolo Turrini\textsuperscript{\rm 1}}\\ 
	\textsuperscript{\rm 1}Department of Computer Science, University of Warwick, United Kingdom \\
\{d.chistikov, grzegorz.lisowski, m.s.paterson, p.turrini\}@warwick.ac.uk
}
\pgfplotsset{compat=1.12}	 

\begin{document}

	\maketitle
	
	\begin{abstract}
		
		We analyse opinion diffusion in social networks, where a finite set of individuals is connected in a directed graph and each simultaneously changes their opinion to that of the majority of their influencers. We study the algorithmic properties of the fixed-point behaviour of such networks, showing that the problem of establishing whether individuals converge to stable opinions is PSPACE-complete.
		
	\end{abstract}

	\section{Introduction}
	
	Social networks are a well established paradigm for the computational analysis of real-world phenomena such as disease spreading  \cite{AIDS,jackson2010social}, product adoption \cite{AptM14,AptMS16} and opinion diffusion \cite{axelrod97,GrandiLP15}. They are typically modelled as directed graphs over a finite set of individuals possessing certain properties, such as opinions, which spread through the network according to predefined rules. For instance, protocols for the spread can be defined as a function of the influencers' opinions.

	In the plethora of social network models, threshold-based ones are certainly the best known. There, agents adopt an opinion if and only if it is shared by a given threshold of the incoming connections. While these models have a long standing tradition in the social sciences, originating from \citeauthor{grano78} (\citeyear{grano78}), they have received revived attention in artificial intelligence, {including contributions of \citeauthor{ferraioli2016decentralized} (\citeyear{ferraioli2016decentralized}), \citeauthor{auletta2017robustness} (\citeyear{auletta2017robustness}) or \citeauthor{bilo2018opinion} (\citeyear{bilo2018opinion}).} Notably, \citeauthor{auletta2018reasoning} have received the IJCAI 2018 Distinguished Paper Award for a study of communication in threshold-based social network models. 
	
	One of the major challenges associated with these models is that convergence of the diffusion protocol is not guaranteed.
	Imagine you would like to have your agents make a collective decision and let them discuss first, agreeing that they would cast their vote once they have made up their mind.
	Depending on the chosen diffusion protocol and the initial distribution of opinions, the process might never terminate. This is the case for synchronous threshold models.
	Clearly, any network will converge for {\em some} initial input, for instance when your agents already think the same to start with. 
	However this is not true in general.
	
	The typical path taken to circumvent the issue is to restrict the analysis to networks that  always converge, as studied by \citeauthor{GrandiLP15} (\citeyear{GrandiLP15}), \citeauthor{bredereck2017manipulating} (\citeyear{bredereck2017manipulating}) and \citeauthor{BotanEtAlCOMSOC2018} (\citeyear{BotanEtAlCOMSOC2018}). Another is to consider specific protocols which guarantee termination, as done for instance by \citeauthor{auletta2018reasoning} (\citeyear{auletta2018reasoning}): they propose an opinion-revision protocol for agents who disagree with a distinguished opinion.
	
	Recently, \citeauthor{christoff2017stability} (\citeyear{christoff2017stability}) have provided a characterisation of networks in which termination of the threshold-based opinion diffusion protocol is guaranteed. However, we still do not know whether characterising convergent networks is of any advantage for their algorithmic analysis,
	in other words, whether we can have a characterisation that is easier to check than actually running the protocol until converging or looping in some way.
	Here, we settle this problem.

	\paragraph{Our contribution.} 
	We study the convergence of opinion diffusion in social networks, modelled as directed graphs over a finite set of individuals, who simultaneously update their opinions. They switch their opinions if and only if the majority  of their influencers  disagrees with them. 
	We look at labelled networks, where individuals start with a binary opinion, and study the problem of whether that network converges. We also look at unlabelled networks and consider the problem of whether a labelling exists for which the network does not converge --- this problem concerns the \emph{structural} aspect of opinion diffusion's convergence. 
	Our contribution is two-fold: firstly, we present some classes of networks which are guaranteed to converge, and secondly we show that the problem of establishing whether a network converges  is PSPACE-complete even for the simplest of such protocols, closing
	a gap in the literature. {In fact, we show that any characterisation of such networks, including the one provided by  \citeauthor{christoff2017stability} (\citeyear{christoff2017stability}) cannot result in an efficient procedure for verifying the convergence of the considered protocol (unless {P=PSPACE}).}  
	
	We emphasize that even though our protocol is relatively simple, the computational complexity lower bounds that we obtain extend directly to more general models. For instance, the PSPACE-hardness of the considered problems lifts to the scenario in which each agent has its own specific update threshold. So our result implies that no complete characterisation of convergent networks can be efficiently computed in practice for a wide range of plausible diffusion protocols.
	
	\paragraph{Related literature}

	Our results have implications for various lines of research using opinion diffusion models.
	
	\begin{description}
		
		\item[\small Social Influence Models] The graph-like structure of social networks has attracted interest in computer science, with  studies of the influence weight of nodes in the network \cite{KempeKT05} and the properties of the influence function \cite{grabisch10}. Social influence has been widely analysed in the social sciences, from the point of view of strategic behaviour \cite{isbell58} and its implications for consensus creation \cite{degroot74} and cultural evolution \cite{axelrod97}.

		\item[\small Opinion Manipulation Models] Issues of convergence are extremely relevant to models that deal with opinion manipulation. For instance, Bredereck and Elkind (\citeyear{bredereck2017manipulating}) study a scenario where an external agent wishes to transform the opinion of a number of members of a network to induce desired fixed-point conditions. Further, control of collective decision-making  \cite{FaliszewskiR16} is an important topic in algorithmic mechanism design: the difficulty of establishing whether manipulation is a real threat is paramount for system security purposes.

		\item[\small Deliberative Democracy and Social Choice] Opinion diffusion underpins recent models of deliberative democracy, in terms of delegation  \cite{list}, representation \cite{EndrissG14}, and stability  \cite{christoff2017stability}. Formal models of democratic representation build on an underlying consensus-reaching protocol  \cite{degroot74,Brill18}.  Social networks have also become of major interest to social choice theory, with propositional opinion diffusion \cite{GrandiLP15} emerging as a framework for social choice on social networks \cite{grandi2017social}.

	\end{description}
	
	\paragraph*{Related computational models.}
	
	If the social networks are modelled as undirected, rather than directed, graphs, it has long been known that convergence takes at most a polynomial number of steps under majority updates \cite{ChaccFP85}. In these models,  PSPACE-hardness results have only been shown for more powerful {\em block sequential} update rules \cite{GolesMST16}.
	
	Convergence is a PSPACE-complete property in various related models, notably directed discrete Hopfield networks \cite{Orponen93} and
	Boolean dynamical systems
	(see, e.g., \citeauthor{BarrettHMRRS03} \citeyear{BarrettHMRRS03} and \citeyear{BarrettHMRRST07}).
	Hardness in these results (and their strengthenings, as studied by \citeauthor{OgiharaU17}  (\citeyear{OgiharaU17}), \citeauthor{RosenkrantzMRS18} (\citeyear{RosenkrantzMRS18}) and \citeauthor{KawachiOU19} (\citeyear{KawachiOU19}))
		crucially depends on the availability of functions
		that \emph{identify} $0$ and~$1$ (see the discussion of the ingredients for the hardness proofs later on).
		Opinion diffusion is instead based on self-dual functions, where flipping all inputs to a self-dual function always leads to flipping its output. In other words, in the setting we consider the diffusion protocol is symmetric with respect to opinions held by agents.

		Whilst \citeauthor{Kosub08} (\citeyear{Kosub08}) shows the NP-completeness of deciding the existence
		of a fixed-point configuration if \emph{all} self-dual functions
		are available, our update rule, in comparison, is monotone (i.e., has no negation).
		Moreover, sparse graphs of bounded fan-in --- with each agent having up to six influencers --- suffice for our proof of PSPACE-hardness.
		In the related model of cellular automata,
		known results show that majority is ``arguably the most interesting" local update rule \cite{Tosic17}.

		\paragraph{Paper structure} We first present our basic setup and examples of networks whose convergence is easy to check. Subsequently we prove that determining convergence is PSPACE-complete. Finally, we conclude by discussing the ramifications of our results and future research directions.

		\section{Opinion Diffusion} \label{sec:preliminaries}
		\paragraph{Social Networks.}	
		Let $N=\{1,2,\ldots, n\}$ be a finite set of agents and $E$ be a simple directed graph over $N$, i.e., an irreflexive relation over the set of agents. We call a tuple $(N,E)$ a {\em social network}.
		The idea is that each agent is influenced by the incoming edges and influences the outgoing ones.
		For each $i\in N$ we define the set $E[i]=\{j \mid (i,j) \in E\}$, 
		i.e., the set of agents that $i$ influences. Similarly, we define the set  $E^{-1}[i]=\{j \mid (j,i) \in E\}$, the {\em influencers of}~$i$.
		
		We are interested in how opinions spread in a social network following the influence relation. For this we equip agents with opinions, giving {\em labelled social networks}.

		\begin{definition}[Labelled Social Network]
			A \emph{labelled social network} is a tuple $\textit{SN} = (N, E, f)$, where:
			
			\begin{itemize}
				\item $(N, E) $ is a social network,
				\item $f: N \rightarrow \{0,1 \} $ is a binary labelling of each node.
			\end{itemize}
		\end{definition}
		
		
		
		Figure \ref{ExConv} gives examples of an unlabelled and a labelled social network.

		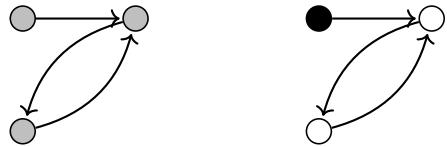
\begin{figure}[t]
			\centering
			
			\begin{tikzpicture}
			[->,shorten >=1pt,auto,node distance=1.5cm,
			semithick]
			\node[shape=circle,draw=black, fill=lightgray] (A) {};
			\node[shape=circle,draw=black, fill=lightgray] (B) [ right of= A] {};
			\node[shape=circle,draw=black, fill=lightgray] (C) [below of = A] { };

			\draw [thick,->]  (A) to(B) ;
			\draw [thick,->, bend right] (B) to [bend right] (C) ;
			\draw [thick,->] (C) to [bend right] (B) ;
			\end{tikzpicture}\hspace{2cm}
			\begin{tikzpicture}
			[->,shorten >=1pt,auto,node distance=1.5cm,
			semithick]
			\node[shape=circle,draw=black, fill=black] (A) {};
			\node[shape=circle,draw=black] (B) [ right of= A] {};
			\node[shape=circle,draw=black] (C) [below of = A] { };

			\draw [thick,->]  (A) to(B) ;
			\draw [thick,->, bend right] (B) to [bend right] (C) ;
			\draw [thick,->] (C) to [bend right] (B) ;
			\end{tikzpicture}
			
			\caption{On the left side, an unlabelled social network. On the right side, one of its labellings. Throughout the paper nodes coloured in black correspond to labelling 1, white nodes to labelling 0 and grey nodes are unlabelled.}
			\label{ExConv}
			
		\end{figure}
		
		\paragraph{Opinion Diffusion Protocol.}

		We model opinion change as an update protocol on the network where each agent $i$ takes the opinion of their influencers, i.e., $E^{-1}[i]$, into account. 
		
		For a given labelled social network $(N,E,f)$, and an agent $i$ let us call $A(i) = \{j \in E^{-1}[i] \mid f(i)=f(j)\}$ the set of influencers who agree with $i$'s opinion, and $D(i) = E^{-1}[i] \setminus A(i)$ the ones who do not.
		
		We assume agents change their opinion if the fraction of their influencers disagreeing with them is (strictly) higher than a half. In particular, a node with $2k$ influencers always takes the opinion of the majority of itself and these influencers. 
		
		\begin{definition}[Opinion Change]
			Let $\textit{SN} = (N, E, f)$ be a labelled social network and $i \in N$ be an agent. Then the \emph{opinion diffusion step} is the function $\textit{OD}: N \rightarrow \{0,1\}$ such that 
			\[ \textit{OD}(\textit{SN}, i) = \begin{cases}  \textit{flip}(f(i))\,\,\,\,\,\, \hfill \mbox{  if }  |D(i)| >  |A(i)| \\  
			f(i) \hfill \mbox{ otherwise} \end{cases} \]
			where $\textit{flip}(k)=1-k$ denotes the change from an original opinion to its opposite value.
		\end{definition}

		We are now ready to define the protocol for the evolution of a labelled social network. Here we focus on the \emph{synchronous update}, in which all agents modify their opinions at the same time. 
		
		\begin{definition}[Synchronous Update]
			Let $\textit{SN} = ( N, E, f )$ be a labelled social network. Then, $\textit{SU}(\textit{SN})= ( N, E, f' )$ is a social network such that for any $i \in N$, $f'(i) = \textit{OD}(\textit{SN}, i)$.
		\end{definition}
		
		The synchronous update protocol is deterministic: given a labelled social network we can compute its state after any given number of synchronous updates. An \emph{update sequence} of a labelled social network \textit{SN} is the infinite sequence of states of \textit{SN} after successive synchronous updates.
		\begin{definition}[Update Sequence]
			Given a labelled social network $\textit{SN} = (N, E, f)$, the \textit{update sequence} generated by \textit{SN} is the sequence of labelled social networks $\textit{SN}_{\textit{us}} = (\textit{SN}_0, \textit{SN}_1\dots )$ such that \textit{SN}$_0$ = \textit{SN} and for every $n \in \mathbb{N}$, $\textit{SN}_{n+1} = \textit{SU}(\textit{SN}_n)$. 
		\end{definition}
		
		For a labelled social network $\textit{SN}$ and agent $i$ we denote by $f_k(i)$ the value given to agent $i$ at time $k$, i.e., at the $k$-th update step. 
		We call a social network $\textit{SN}$  {\em stable} if $\textit{SU}(\textit{SN})=\textit{SN}$. A social network is {\em convergent} if its update sequence contains a stable social network, i.e., if its update sequence reaches a fixed point, its {\em limit network}.
		
		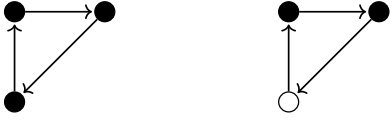
\begin{figure}[t]
			\begin{subfigure}{.200\textwidth}
				\centering
			\scalebox{0.8}{
					\begin{tikzpicture}
				[->,shorten >=1pt,auto,node distance=1.5cm,
				semithick]
				\node[shape=circle,draw=black, fill=black] (A) {};
				\node[shape=circle,draw=black, fill=black] (B) [ right of= A] {};
				\node[shape=circle,draw=black, fill=black] (C) [below of = A] {};

				\draw [thick,->] (A) -- (B) ;
				\draw [thick,->] (B) -- (C) ;
				\draw [thick,->] (C) -- (A) ;

				\end{tikzpicture}}
			\end{subfigure}
			\begin{subfigure}{.200\textwidth}
				
				\centering
				\scalebox{0.8}{
				\begin{tikzpicture}
				[->,shorten >=1pt,auto,node distance=1.5cm,
				semithick]
				\node[shape=circle,draw=black, fill=black] (A) {};
				\node[shape=circle,draw=black, fill=black] (B) [ right of= A] {};
				\node[shape=circle,draw=black] (C) [below of = A] {};

				\draw [thick,->] (A) -- (B) ;
				\draw [thick,->] (B) -- (C) ;
				\draw [thick,->] (C) -- (A) ;

				\end{tikzpicture}}
			\end{subfigure}
			\caption{On the left, a convergent (labelled) social network. On the right, a non-convergent one.}
			\label{FirstEx}
		\end{figure}

		\section{Graph Restrictions}\label{sec:graphs}

		Some networks converge for all initial labellings, while others converge for just some labellings. The lefthand network in Figure \ref{ExConv}, for example, converges for every labelling. However, the social networks displayed in Figure \ref{FirstEx} behave differently. Here we look at specific instances of social networks which converge for every labelling.
		
		
		Let us start with  DAGs, i.e., \emph{directed acyclic graphs}. 
		
		\begin{proposition}\label{acyclic}
			Let \textit{SN}=$(N,E)$ be a DAG. Then \textit{SN} converges in at most $k$ steps for every labelling $f$, where $k$ is the length of the longest path.
		\end{proposition}
		\begin{proof} 
			{
			Given a DAG \textit{SN}=$(N,E)$, consider an arbitrary labelled social network $\textit{SN}'=(N,E, f)$.  Let us write $i \rightarrow j$ for $j \in E[i]$. Since \textit{SN} is acyclic, for every $i \in N$ there is a path to $i$ from some source node of \textit{SN}. 
			Let $\textit{level}(i)$ be the length of the longest such path.
			We will show by induction on $\textit{level}(i)$ that every $f(i)$ will stabilise after at most $\textit{level}(i)$ updates. 
			}
			
			{
			If $\textit{level}(i)$ is 0 then $i$ is a source node and therefore never changes. 
			Suppose that all $i$ such that $\textit{level}(i)=r$ have stabilised after $r$ updates. Take any node $i$ with $\textit{level}(i)=r+1$. Since \textit{SN} is acyclic, for any $n'\in N$ such that  $n'\rightarrow i$, we have $\textit{level}(n') \leq r$. This means $n'$ is already stable after $r$ updates. Hence, $i$ will stabilise within one step after all its influencers have stabilised, i.e., after at most $r+1$ updates.   
			}
		\end{proof}

		Networks that are not DAGs do not always converge, as shown in Figure \ref{FirstEx}. But some of these have interesting properties with respect to convergence. For example cliques, i.e., networks $\textit{SN}=(N,E)$ with $E=N^2\setminus \{(i,i) \mid i \in N\}$. 
		
		\begin{proposition}\label{full} 
			Let \textit{SN}=$(N,E)$ be a clique. \textit{SN} converges for every labelling if and only if $|N|$ is odd. Moreover, if $\textit{SN}$ converges, then it does so after a single update step.
		\end{proposition}
		\begin{proof}
			It is easy to check that if $\textit{SN}$ is evenly split (and therefore of even size) then every agent flips at each update step.  Otherwise, after one update every agent has the opinion of the initial majority.
		\end{proof}

		As checking whether a social network is a clique can be achieved by just counting its edges, the result above shows that for some structures establishing convergence is immediate.
		


		Consider now the  {\em strongly connected components} (SCCs) of a social network, i.e., subgraphs that have a path from each node to every other node and are maximal with respect to set inclusion.
		As is well-known (see e.g., \citeauthor{Bollobas1998} \citeyear{Bollobas1998}), each network $\textit{SN}=(N,E)$  can be partitioned into SCCs, yielding a DAG $\textit{SCC}_{\textit{SN}} = (\textit{SCCs}, E')$ where: (i) $\textit{SCC}s$ is the set of all SCCs of \textit{SN}; (ii) for any $\textit{SCC}_u, \textit{SCC}_v \in \textit{SCCs}$, $(\textit{SCC}_u, \textit{SCC}_v) \in E'$ iff for some $i \in SCC_u, j \in SCC_v$ we have that $j \in E[i]$.  Recall, that the set of SCCs of $\textit{SN}$ can be computed in linear time in the size of $\textit{SN}$.
		
		One might expect that if we knew that each SCC always converges then so would the whole network, or, put otherwise, that every network that always converges will also do so when only influenced by a network that itself always converges. Remarkably, this is not true even for very simple cases, as exemplified in Figure~\ref{fig:SCCs}.

		We now move on to the problem of checking convergence in an arbitrary social network.

              \begin{figure}[t]
			\centering
			\scalebox{0.8}{
			\begin{tikzpicture}
			[->,shorten >=1pt,auto,node distance=1.5cm,
			semithick]
			\node[shape=circle,draw=black, fill=white] (A) {};
			\node[shape=circle,draw=black, fill=black] (B) [ right of= A] {};
			\node[shape=circle,draw=black, fill=black] (C) [right of = B] { };
			\node[shape=circle,draw=black, fill=black] (D) [above of = A] { };
			\node[shape=circle,draw=black, fill=white] (E) [above of = B] { };
			\node[shape=circle,draw=black, fill=white] (F) [above of = C] { };
			
			\node[draw,  dashed, fit=(A) (C)  (B) ](FIt1) {};
			\node[draw,  dashed, fit=(D) (E) (F)](FIt1) {};
						
			\draw [thick,->]  (A) to(B) ;
			\draw [thick,->]  (A) to [bend right](C) ;
			\draw [thick,->]  (B) to [bend right](A) ;
			\draw [thick,->] (B) to (C) ;
			\draw [thick,->] (C) to [bend right] (B) ;
			\draw [thick,->]  (D) to(E) ;
			\draw [thick,->]  (D) to [bend right](F) ;
			\draw [thick,->]  (E) to [bend right](D) ;
			\draw [thick,->] (E) to (F) ;
			\draw [thick,->] (F) to [bend right] (E) ;
			\draw [thick,->] (E) to  (B) ;
			\end{tikzpicture}}
			\caption{A labelled network that does not converge, whose two SCCs (marked by rectangles) 
			do converge for every initial labelling. The convergence of the SCC in the lower tier is influenced 
			by the incoming edge from the upper SCC. }
			\label{fig:SCCs}
		\end{figure}
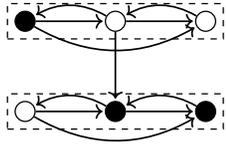

		\section{The Complexity of Checking Convergence} \label{sec:complexity}

		We consider two computational problems with respect to the protocol we are considering. The first of them is checking the convergence of a given labelled social network. 
		
		\begin{quote}
			\noindent \textsc{Convergence}:\\
			\hspace*{-1em} \indent\textit{Input:} Social network $\textit{SN}=( N, E)$ and labelling $f$.\\
			\hspace*{-1em}\textit{Output:} Does \textit{SN} converge from $f$?
		\end{quote}
		
		The second is checking for an unlabelled network whether there is a labelling which \emph{does not} converge.
		\begin{quote}
			\textsc{Convergence guarantee}:\\
			\hspace*{-1em}\textit{Input:} Social network $\textit{SN}=( N, E)$.\\
			\hspace*{-1em}\textit{Output:} Is there a labelling of \textit{SN} from which \textit{SN} does not converge? 
		\end{quote}
		
		
		In the remainder of this section we will prove theorems associated with these two computational problems. 
		\begin{theorem}
			\label{th:fixed-term}
			\textsc{Convergence} is PSPACE-complete.
		\end{theorem}
		
		\begin{theorem}
			\label{th:quant-term}
			\textsc{Convergence guarantee} is PSPACE-complete.
		\end{theorem}

		
		It is important to note that the lower bounds apply to all opinion diffusion models for which our protocol is a special case. In particular, this holds in models with agent-dependent update thresholds or
		with weighted trust levels (i.e., with weighted majority instead of majority).
		
		Let us notice that both problems belong to PSPACE, because each labelling of a social network
		$\textit{SN}=(N,E)$ takes $|N|$ bits and the synchronous update mapping \textit{SU}
		can be evaluated in polynomial time.
		
		The hardness proof of Theorem~\ref{th:fixed-term} can be developed separately,
		but we choose to give a uniform presentation and derive hardness of both problems
		from the same construction, in order to make the proof of Theorem~\ref{th:quant-term}
		easier to follow.
		
		\subsection{Ingredients for the hardness proofs}\label{sec:comp:ing}
		
		The main technical challenge for the hardness proof is that the update mapping
		\textit{SU} applies a self-dual Boolean function (majority):
		if a node and all its influencers flip their opinion, then, after the update,
		the node will have the flipped value. This means, informally, that
		the nodes are indifferent to the identity of $0$ and $1$, which makes a direct simulation
		of propositional logic impossible.
		
		Our construction below is, in hindsight, reminiscent of the observation
		that the \emph{negation of the 3-input majority} is a basis for the class
		of all self-dual functions
		\cite{Post-book}; see, e.g., \cite[Theorem~3.2.3.2]{Lau-book}.
		Our proof, however, does not rely on any advanced topics in the theory of
		Boolean functions and their clones/closed classes.
		
		\paragraph*{Propositional logic and dual rail encoding.}
		
		{Let us introduce the basic technical notions appearing in the proofs of hardness of the considered problems.} We will use Boolean circuits from computational complexity theory. Due to space constraints we omit the detailed introduction of Boolean circuits, which can be found, e.g., in \cite[section 4.3]{PapadimitriouCC}.
		Signals in these circuits are Boolean values, true and false, and we will encode them
		in our social networks. We need to encode logical gates (\textsc{and} and \textsc{not}) and constant gates (\textsc{true} and \textsc{false}) too.
		
		 {We use the dual rail encoding due to the monotonicity of the opinion diffusion protocol. Indeed, in the current setting opinions reinforce themselves, so logical negation cannot be directly simulated.}
		 
		In the dual rail encoding, instead of considering individual nodes in a social network,
		we will be often considering related pairs of nodes, called \emph{dual pairs}.
		The two nodes in a dual pair are ordered.
		Given a labelling of the network,
		a dual pair is \emph{valid} if its two nodes disagree, i.e., take different values,
		and \emph{invalid} otherwise.
		Dual pairs will be building blocks in our construction,
		and our network will have a mechanism to ensure their validity. 
		
		Our first step is to build constant gates.
		We introduce a distinguished dual pair, the \emph{base pair};
		as long as it is valid, we assume without loss of generality that its
		two nodes have values $(1, 0)$.
		There is only one base pair in the network.
		Now for every valid dual pair in the network, we interpret $(1,0)$ as \emph{true} and $(0, 1)$ as \emph{false}.
		
		The next step is to build logical gates.
		All these gates in our circuits have fan-in $1$ or $2$, that is, each gate receives
		input from at most $2$ other gates. The gates are depicted in Figures 4 and 5 and described in Example 1.
		
		\begin{figure}[h]\label{AND}
			\centering
			\scalebox{0.8}{
			\begin{tikzpicture}
			[->,shorten >=1pt,auto,node distance=1.6cm,
			semithick]
			\node[shape=circle,draw=black, fill=lightgray] (A) {};
			\node[shape=circle,draw=black, fill=lightgray] (B) [ right of= A] {};
			\node[draw, ellipse, dashed, fit=(A) (B) ](FIt1) {};
			
			\node[shape=circle,draw=black, fill=lightgray] (C) [left of = A] {};
			\node[shape=circle,draw=black, fill=lightgray] (D) [left of = C] {};
			\node[draw, ellipse, dashed, fit=(C) (D) ](FIt1) {};

			\node[shape=circle,draw=black] (E) [below of = C] {};
			\node[shape=circle,draw=black, fill=black] (F) [below of = D] {};
			\node[draw, double, fit=(E) (F) ](FIt1)(C5)  {};
			
			\node[shape=circle,draw=black, fill=lightgray] (G) [below of = B] {};
			\node[shape=circle,draw=black, fill=lightgray] (H) [below of = A] {};
			\node[draw, ellipse, dashed, fit=(G) (H) ](FIt1) {};

			\draw [thick,->]  (A) to(H) ;
			\draw [thick,->]  (B) to(G) ;
			
			\draw [thick,->]  (D) to(H) ;
			\draw [thick,->]  (B) to(G) ;
			\draw [thick,->]  (C) to(G) ;
			
			\draw [thick,->]  (E) to(H) ;
			\draw [thick,->]  (F) to[bend right](G) ;
			
			\end{tikzpicture}}
			\caption{The \textsc{and} gadget.}
			\label{fig:AND}
		\end{figure}
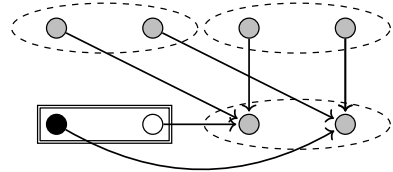

		\begin{figure}[h]\label{NOT}
			\centering
			\scalebox{0.8}{
			\begin{tikzpicture}
			[->,shorten >=1pt,auto,node distance=1.5cm,
			semithick]
			\node[shape=circle,draw=black, fill=lightgray] (A) {};
			\node[shape=circle,draw=black, fill=lightgray] (B) [ right of= A] {};
			\node[draw, ellipse, dashed, fit=(A) (B) ](FIt1) {};

			\node[shape=circle,draw=black, fill=lightgray] (G) [below of = B] {};
			\node[shape=circle,draw=black, fill=lightgray] (H) [below of = A] {};
			\node[draw, ellipse, dashed, fit=(G) (H) ](FIt1) {};

			\draw [thick,->]  (A) to(G) ;
			\draw [thick,->]  (B) to(H) ;

			\end{tikzpicture}}\hspace{1.5cm}
			\scalebox{0.8}{
			\begin{tikzpicture}
			[->,shorten >=1pt,auto,node distance=1.5cm,
			semithick]
			\node[shape=circle,draw=black, fill=lightgray] (A) {};
			\node[shape=circle,draw=black, fill=lightgray] (B) [ right of= A] {};
			\node[draw, ellipse, dashed, fit=(A) (B) ](FIt1) {};

			\node[shape=circle,draw=black, fill=lightgray] (G) [below of = B] {};
			\node[shape=circle,draw=black, fill=lightgray] (H) [below of = A] {};
			\node[draw, ellipse, dashed, fit=(G) (H) ](FIt1) {};

			\draw [thick,->]  (A) to(H) ;
			\draw [thick,->]  (B) to(G) ;

			\end{tikzpicture}}
			\caption{\textsc{not} gadget on the left,  \textsc{nop} gadget on the right.}
			\label{fig:NOT}
		\end{figure}
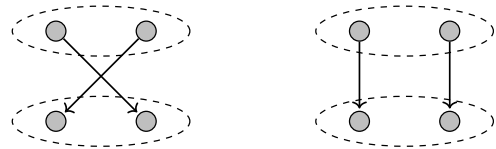

		\begin{example}
			\label{ex:gates}
			The gadget in Figure~\ref{fig:AND} models an \textsc{and} gate, and
			the gadget in Figure~\ref{fig:NOT} (left) models a \textsc{not} gate. 
			The \textsc{and} gadget relies on the base pair, which is depicted
			as a double rectangle.
			In more detail, if at time $t$ the input dual pairs (the two upper ovals) in the \textsc{and} gadget are valid,
			then at time $t+1$ the output dual pair is valid and represents the \textsc{and} of the two
			input values (and similarly for the \textsc{not} gate).
			Finally, the gadget in Figure~\ref{fig:NOT} (right) models a \textsc{nop} (no-operation) gate: at time $t+1$ the output pair is a copy of the input pair at time $t$.
		\end{example}
		
		\paragraph*{Turing machines.}
		
		Further in our reduction we will need to build Boolean circuits
		to simulate the behaviour of Turing machines.
		
		We will describe a restricted version of Turing machines that we
		use to prove Theorems~\ref{th:fixed-term} and~\ref{th:quant-term}.
		These Turing machines are \emph{polynomially space-bounded},
		or \emph{PSPACE machines} (referring to the complexity class);
		see, e.g., \cite[section 4.2]{AroraBarak} and \cite[chapter 19]{PapadimitriouCC} for a more detailed discussion.

		We will not need a formal definition of Turing machines in this paper
		and will instead rely on the following properties only:
		\begin{enumerate}
			\item Any Turing machine has a finite description.
			\item Any Turing machines can be \emph{run} on arbitrary input strings of arbitrary length $m \ge 0$
			over a fixed finite alphabet.
			\item At any point during a run, an instantaneous description of a Turing machine $M$
			(a \emph{configuration})
			can be encoded by a bit string of length $c \cdot m^d$,
			where the constants $c$ and $d$ depend only on the machine $M$.
			\item A Turing machine may either \emph{halt} at some point during the run, or
			\emph{diverge} (run forever).
			\item A run is a finite or infinite sequence of configurations;
			each configuration is either \emph{halting} or has a unique \emph{successor} configuration.
		\end{enumerate}
		%
		We will identify configurations of Turing machines with their encodings
		as $n$-bit strings (strings of truth values). Here $n = c \cdot m^d$;
		when $m$ is fixed, $n$ is the same in all possible configurations.
		
		For a given $n$, we will assume for the sake of simplicity
		that all $n$-bit strings represent valid configurations.
		This assumption does not invalidate our reduction
		and can in fact be eliminated using the technique of the following lemma.
		
		\begin{lemma}
			\label{l:step}
			Given a Turing machine $M$ and an integer $n \ge 1$,
			there exists an acyclic social network \textit{SN} with the following properties:
			\begin{itemize}
				\item \textit{SN} contains the base pair and has $2 n$ further sources and $2 n$ sinks, grouped into $n$ and $n$ dual pairs;
				\item every path from a source to a sink has the same length $h$, independent of $n$;
				\item \textit{SN} \emph{simulates} $M$: if at time $t$ the base pair and input dual pairs are valid
				and represent a configuration $s^{(0)} \in \{0, 1\}^n$,
				then at time $t+h$ if $s^{(0)}$ is non-halting the output dual pairs are valid and represent $s^{(1)}$, 
				the successor configuration of $s^{(0)}$;
				otherwise at least one output dual pair at time $t+h$ is invalid;
				\item \textit{SN} can be constructed in time polynomial in~$n$ and in the description of $M$.
			\end{itemize}
		\end{lemma}
		\begin{proof}
			The assertion relies on the observation
			(following the lines of \cite[Theorem 6.6]{AroraBarak}, or \cite[section 8.2]{PapadimitriouCC})
			that
			for every polynomially space-bounded
			Turing machine $M$ and every integer $n$,
			there exists a Boolean circuit which:
			\begin{itemize}
				\item has $n$ inputs and $n$ outputs,
				\item has equal-length paths from inputs to outputs (where this length $h$ is independent of $n$),
				\item transforms an arbitrary non-halting configuration of $M$ into its successor configuration.
				\item can be constructed in time polynomial in~$n$ and in the description of $M$.
			\end{itemize}
			These properties map into the assertions of the lemma, using dual pairs
			as nodes in the circuit, and \textsc{and} and \textsc{not} gadgets from Example~1 as gates.
			To make the network satisfy the second assertion of the lemma, we extend it using
			\textsc{nop} gadgets where necessary.
		\end{proof}

		\paragraph*{Fuse line, valve, and alarm.}
		We will need a mechanism to check initial validity of dual pairs in our construction,
		as well as to detect the halting of a Turing machine, following Lemma~\ref{l:step}.
		If a dual pair is or becomes invalid,
		this will force the convergence of the social network.

		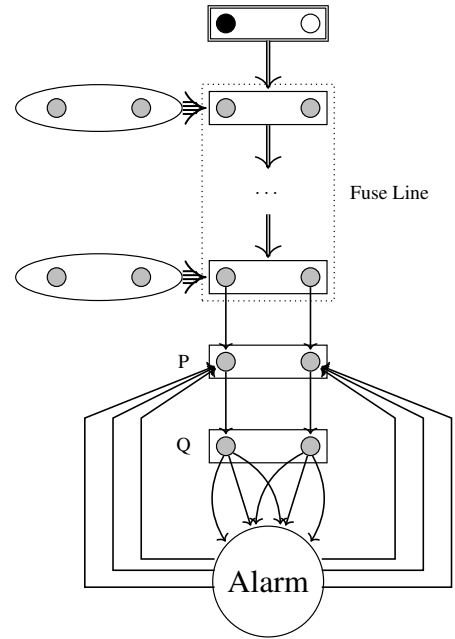
\begin{figure}[t]
			\centering
			\scalebox{0.75}{
			\begin{tikzpicture}
			[->,shorten >=1pt,auto,node distance=1.5cm,
			semithick]
			\node[shape=circle,draw=black, fill=lightgray] (A) {};
			\node[shape=circle,draw=black, fill=lightgray] (B) [ right of= A] {};
			\node[draw, fit=(A) (B) ](FIt1)(M) {};

			\node[left of=M](LA) {Q};
			
			
			\node[shape=circle,draw=black, fill=lightgray] (E) [above of = A]{};
			\node[shape=circle,draw=black, fill=lightgray] (F) [ above of = B] {};
			\node[draw, fit=(E) (F) ](FIt1)(O) {};
			\draw [thick, ->]  (E) to(A) ;
			\draw [thick, ->]  (F) to(B) ;
			\node[left of=O](LB) {P};

			\node[shape=circle,draw=black, fill=lightgray] (G1) [above of = E]{};
			\node[shape=circle,draw=black, fill=lightgray] (G2) [ right of = G1] {};
			\node[draw, fit=(G1) (G2) ](FIt1)(G3) {};
			\draw [thick, ->]  (G1) to(E) ;
			\draw [thick, ->]  (G2) to(F) ;

			\node[shape=circle,draw=white] (Y) [above of = G3] {$\dots$};
			\draw [thick,double, ->]  (Y) to(G3) ;

			\node[shape=circle,draw=black, fill=lightgray] (B1) [left of = G1]{};
			\node[shape=circle,draw=black, fill=lightgray] (B2) [left of= B1] {};
			\node[draw, ellipse, fit= (B1) (B2) ](FIt1)(C1) {};
			\triplearrow{arrows={-Implies}}{(C1) -- (G3)}

			\node[shape=circle,draw=black, fill=lightgray] (D1) [above of = Y, above of=G1]{};
			\node[shape=circle,draw=black, fill=lightgray] (D2) [ right of= D1] {};
			\node[draw, fit=(D1) (D2) ](FIt1)(D3) {};
			\draw [thick,double, ->]  (D3) to(Y) ;

			\node[shape=circle,draw=black, fill=lightgray] (E1) [left of = D1]{};
			\node[shape=circle,draw=black, fill=lightgray] (E2) [left of= E1] {};
			\node[draw, ellipse, fit= (E1) (E2) ](FIt1)(E3) {};
			\triplearrow{arrows={-Implies}}{(E3) -- (D3)}

			\node[shape=circle,draw=black, fill=black] (A5) [ above of=D1]{};
			\node[shape=circle,draw=black] (A6) [ right of= A5] {};
			\node[draw, double, fit=(A5) (A6) ](FIt1)(C5)  {};
			\draw [thick,double, ->]  (C5) to(D3) ;

			\node[shape=circle,draw=black, scale=1.6] (X) [below of = M] {Alarm};
			
			\node[draw, dotted, fit=(D3) (Y) (G3) ](FIt1)(FUSE) {};
			\node[shape=circle, ] (LABELFUSE) [right of = FUSE] {\textcolor{white}{fafascdsc}Fuse Line};
			
			\draw [thick, ->]  (A) to[bend left] (X) ;
			\draw [thick, ->]  (A) to [bend right](X) ;
			\draw [thick,->]  (B) to(X) ;
			\draw [thick,->]  (B) to [bend left](X) ;
			\draw [thick,->]  (B) to [bend right](X) ;
			\draw [thick, ->]  (A) to (X) ;

			
			\draw [thick, ->]  (-0.2,-2.5) to (-2.5,-2.5) to (-2.5, 0.5) to (E);
			\draw [thick, ->]  (-0.2,-2.2) to (-2,-2.2) to (-2, 0.5) to (E);
			\draw [thick, ->]  (-0.1999,-2) to (-1.5,-2) to (-1.5, 0.5) to (E);

			\draw [thick, ->]  (1.7,-2.5) to (4, -2.5) to (4, .5) to (F);
			\draw [thick, ->]  (1.7,-2.2) to (3.5, -2.2) to (3.5, .5) to (F);
			\draw [thick, ->]  (1.7,-2.) to (3, -2.) to (3, .5) to (F);

			\end{tikzpicture}}
			\caption{The fuse line.}
			\label{fig:fuse}
		\end{figure}

		\begin{figure}[ht]
			\centering
			\scalebox{0.8}{
			\begin{tikzpicture}
			[->,shorten >=1pt,auto,node distance=1.1cm,
			semithick]
			\node[shape=circle,draw=black, fill=lightgray] (A) {};
			\node[shape=circle,draw=black, fill=lightgray] (B) [ right of= A] {};
			\node[draw, dashed, fit=(A) (B) ](FIt1) {};

			\node[shape=circle,draw=black, fill=lightgray] (G) [below of = B] {};
			\node[shape=circle,draw=black, fill=lightgray] (H) [below of = A] {};
			\node[draw, dashed, fit=(G) (H) ](FIt1) {};

			\draw [thick,->]  (A) to(G) ;
			\draw [thick,->]  (A) to(H) ;
			\draw [thick,->]  (B) to(H) ;
			\draw [thick,->]  (B) to(G) ;

			\end{tikzpicture}
			\hspace{2em}
			\begin{tikzpicture}
			[->,shorten >=1pt,auto,node distance=1.1cm,
			semithick]
			\node[shape=circle,draw=black, fill=lightgray] (A) {};
			\node[shape=circle,draw=black, fill=lightgray] (B) [ right of= A] {};
			\node[draw, fit=(A) (B) ](FIt1)(M) {};

			\node[shape=circle,draw=black, fill=lightgray] (G) [below of = B] {};
			\node[shape=circle,draw=black, fill=lightgray] (H) [below of = A] {};
			\node[draw, fit=(G) (H) ](FIt1) (N){};

			\draw [thick, double, ->]  (M) to(N) ;

			\end{tikzpicture}}
			\caption{Two pairs in the fuse line, one feeding into the other.\\
				Left: in detail. Right: simplified drawing (corresponding to connections between pairs in the fuse line as depicted in Figure~\ref{fig:fuse}), abbreviating the connections in the left picture. }
			\label{fig:fuse-feed}
		\end{figure}
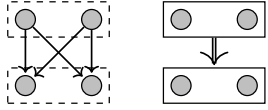
		
		\begin{figure}[ht]
			\centering
			
			\begin{subfigure}[t]{0.4\linewidth}
			\scalebox{0.8}{\begin{tikzpicture}
			[->,shorten >=1pt,auto,node distance=1.1cm,
			semithick]
			\node[shape=circle,draw=black, fill=lightgray] (A) {};
			\node[shape=circle,draw=black, fill=lightgray] (B) [ right of= A] {};
			\node[draw, ellipse, dashed, fit=(A) (B) ](FIt1) (P){};
			
			\node[shape=circle,draw=black, fill=lightgray] (I) [below of= A, left of=A] {};
			\node[shape=circle,draw=black, fill=lightgray] (J) [ right of= I] {};
			\node[shape=circle,draw=black, fill=lightgray] (K) [right of =J]{};
			\node[shape=circle,draw=black, fill=lightgray] (L) [ right of= K] {};

			\node[shape=circle,draw=black, fill=lightgray] (G) [below of = J] {};
			\node[shape=circle,draw=black, fill=lightgray] (H) [below of = K] {};
			\node[draw, dashed, fit=(G) (H) ](FIt1) (N){};

			\draw [thick,->]  (A) to(I) ;
			\draw [thick,->]  (A) to(J) ;
			\draw [thick,->]  (B) to(K) ;
			\draw [thick,->]  (B) to(L) ;
			
			\draw [thick,->]  (I) to(G) ;
			\draw [thick,->]  (I) to(H) ;
			
			\draw [thick,->]  (J) to(G) ;
			\draw [thick,->]  (J) to(H) ;
			
			\draw [thick,->]  (K) to(G) ;
			\draw [thick,->]  (K) to(H) ;
			
			\draw [thick,->]  (L) to(G) ;
			\draw [thick,->]  (L) to(H) ;
			
			\end{tikzpicture}}
			\end{subfigure}
			~
		\begin{subfigure}[t]{0.4\linewidth}
			\scalebox{0.8}{
			\begin{tikzpicture}
			[->,shorten >=1pt,auto,node distance=1.5cm,
			semithick]
			\node[shape=circle,draw=black, fill=lightgray] (A) {};
			\node[shape=circle,draw=black, fill=lightgray] (B) [ right of= A] {};
			\node[draw, ellipse, fit=(A) (B) ](FIt1)(M) {};

			\node[shape=circle,draw=black, fill=lightgray] (G) [below of = B] {};
			\node[shape=circle,draw=black, fill=lightgray] (H) [below of = A] {};
			\node[draw, fit=(G) (H) ](FIt1) (N){};

			\triplearrow{arrows={-Implies}}{(M) -- (N)}

			\end{tikzpicture}}
			\end{subfigure}
			\caption{Dual pair connected to pair from the fuse line.\\
				Left: in detail. Note how the influence of the input pair on the output pair is stronger than in Figure~\ref{fig:fuse-feed}. Right: simplified drawing (used in Figure~\ref{fig:fuse}), abbreviating the connections in the left picture.}
			\label{fig:fuse-connect}
		\end{figure}
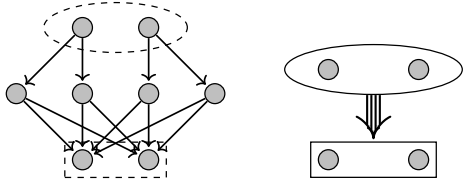
		
		The mechanism consists of a \emph{fuse line} (sequence of pairs of nodes)
		leading to a \emph{valve} and \emph{alarm} (an even clique), as shown in Figure~\ref{fig:fuse}.

		Let us first discuss the fuse line itself.
		Pairs of nodes in the fuse line are depicted by rectangles. 
		Each pair in the fuse line (except for the last) feeds into the succeeding pair as shown in Figure~\ref{fig:fuse-feed}.
		In addition, all other dual pairs in the entire network (depicted for the sake of clarity as ovals) will also connect to distinct pairs in the fuse line
		as shown in Figure~\ref{fig:fuse-connect}.
		We will not think of the pairs in the fuse line as dual pairs.
		
		At the end of the fuse line
		shown in Figure~\ref{fig:fuse},
		the big circle is a clique of $2 k$ nodes (an \emph{alarm}), $k \ge 2$,
		and the \emph{valve} mechanism is formed by the two rectangles (pairs) $P$, $Q$, and the alarm.
		Both nodes of pair $Q$ have edges to each node in the alarm, and
		all nodes in the alarm have edges to both nodes of pair $P$. In the following analysis, we say that the alarm is
		\emph{evenly split} if exactly $k$ of its nodes are labelled $0$.
		We say that the alarm \emph{goes off} at time $t$
		if all of its nodes agree at this time
		(we will usually imply that this was not the case at time $t-1$).

		We show now several properties of this network which will be crucial for the PSPACE-hardness reduction.
		\begin{lemma}
			\label{l:fuse-step}
			If at time $t \ge 1$ a pair in the fuse line is invalid,
			then:
			\begin{enumerate} [label=\textup{(\alph*)}]
				\item\label{l:fuse-step:remain} it remains invalid forever and
				\item\label{l:fuse-step:successor} the succeeding pair is invalid from time $t+1$ on.
			\end{enumerate}
		\end{lemma}
		
		\begin{proof}
			Assertion~\ref{l:fuse-step:remain} follows the fact that
			the two nodes in any single pair in the fuse line
			have the same set of influencers.
			\par
			In order for assertion~\ref{l:fuse-step:successor} to fail,
			the succeeding pair must be valid at times $t$ and $t+1$.
			Since, again, the two nodes in this succeeding pair
			have the same set of (six) influencers, this set should be
			evenly split at time~$t$. But this is impossible,
			because two of these influencers agree
			by the assumption of the lemma, and the remaining four
			cannot be split into $1$ and $3$ for every $t \ge 1$
			by the construction of the connection in Figure~\ref{fig:fuse-connect}.
		\end{proof}
		
		\begin{lemma}
			\label{l:stuck}
			If some dual pair in the network is invalid at some time,
			then the last pair in the fuse line becomes invalid at some time
			and remains invalid forever.
		\end{lemma}
		
		\begin{proof}
			Follows directly from Lemma~\ref{l:fuse-step}.
		\end{proof}
		
		The final part of our construction of the network is that
		every node in the alarm has edges to every node in the network,
		except for the fuse line, nodes connecting dual pairs to pairs in the fuse line and pair $Q$ of the valve
		(in other words, to all dual pairs and to pair $P$).
		This includes the two nodes of the base pair (depicted, as previously,
		as a double rectangle).

		\begin{lemma}
			\label{l:alarm}
			Suppose at time $t$ at least one of the following conditions holds:
			\begin{enumerate}  [label=\textup{ (\alph*)}]
				\item\label{l:alarm:fuse} the last pair in the fuse line is invalid,
				\item\label{l:alarm:p} the two nodes of $P$ agree,
				\item\label{l:alarm:q} the two nodes of $Q$ agree, or
				\item\label{l:alarm:alarm} the alarm is not evenly split.
			\end{enumerate}
			Then by time $t+3$ the alarm goes off
			and by time $t+6$ all nodes in the network agree.
		\end{lemma}
		
		\begin{proof}
			Suppose that at some time $t'$ the alarm is split into sets of size
			$m \le k$ and $2 k - m$.
			If $m \le k-2$, the influence of $Q$ on the alarm is negligible,
			so the alarm goes off at time $t'+1$.
			If $m = k$, then all nodes in the alarm will flip at time $t'+1$
			if the two nodes of $Q$ disagree; otherwise the alarm goes off.
			Finally, if $m = k-1$, then all nodes in the alarm will flip at time $t'+1$
			if the two nodes of $Q$ agree and side with the minority, otherwise
			the alarm goes off.
			
			We now show that, under the conditions of the lemma, the alarm
			will necessarily go off at some time $\le t+3$.
			%
			If this does \emph{not} happen, then, by the argument above,
			either (\textit{i}) the alarm remains evenly split (and the nodes
			in each of the pairs $P$ and $Q$ disagree),
			or (\textit{ii}) the alarm is split into sets of size $k-1$ and $k+1$,
			flipping on each step, and the nodes of $Q$ keep agreeing with each
			other and alternating
			between $(0,0)$ and $(1,1)$.
			
			Consider scenario~(\textit{i}).
			Note that cases~\ref{l:alarm:p} and~\ref{l:alarm:q} are incompatible
			with this scenario, because $Q$ copies $P$ and influences
			all nodes in the alarm.
			Case~\ref{l:alarm:alarm} is not compatible
			with this scenario either. So only case~\ref{l:alarm:fuse} remains.
			But since in this scenario
			the alarm is evenly split, the pair $P$ will simply copy the last pair
			of the fuse line, and at time $t+1$ we are essentially in case~\ref{l:alarm:p}.
			So, in scenario~(\textit{i}), the alarm will go off by time $t+3$.
			
			Let us now consider scenario~(\textit{ii}).
			Assume with no loss of generality that at time $t$
			the alarm has $k-1$ nodes labelled $0$ and $k+1$ nodes labelled $1$,
			and that $Q$ is labelled
			$(0, 0)$ at this time (siding with the minority).
			Note that the pair $Q$ keeps flipping; and
			since the alarm is split into sets of size $k-1$ and $k+1$,
			the nodes of $P$ simply copy the majority in the alarm.
			Therefore,
			the labellings must follow the following diagram:
			\begin{equation*}
			\begin{array}{lccc}
			&  t    &  t+1  &  t+2  \\
			P                        &       & (1,1) &       \\
			Q                        & (0,0) & (1,1) & (0,0) \\
			\text{$k-1$ in the alarm}&    0  &  1    &  0    \\
			\text{$k+1$ in the alarm}&    1  &  0    &  1    
			\end{array}
			\end{equation*}
			But this labelling of $Q$ at time $t+2$ is not possible,
			because the pair $Q$ simply copies the pair $P$.
			Therefore, $Q$ will remain at $(1,1)$ instead;
			and so at time $t+3$ the alarm will go off.
			
			It remains to prove that, in all cases,
			in at most $3$ steps from
			the alarm going off,
			all nodes in the network will agree.
			Since all $2 k$ nodes in the alarm influence all dual pairs in the network,
			and the indegree of each node in every dual pair is at most $3$
			(not counting the edges from the alarm), the influence of the alarm
			will prevail as long as $k \ge 2$,
			i.e., all dual pairs
			will become invalid and assume this value by time $t+4$.
			All pairs in the fuse line will follow by time $t+6$.
			At the same time, pairs $P$ and $Q$ of the valve will follow the alarm
			no later than at times $t+4$ and $t+5$, respectively.
			This completes the proof.
		\end{proof}

		\newcommand{\newextmathcommand}[2]{%
			\newcommand{#1}{\ensuremath{#2}\xspace}
		}
		
		\newextmathcommand{\MainNet}{\textit{MN}}
		\newextmathcommand{\mainLabelling}{f}
		\newextmathcommand{\auxLabelling}{a}
		\newextmathcommand{\someLabelling}{g}

		\paragraph*{Auxiliary labelling~$\auxLabelling(s)$.}
		
		Let \textit{SN} be the social network from Lemma~\ref{l:step}
		and $s \in \{0, 1\}^n$ a configuration.
		Recall that \textit{SN} is acyclic and all paths from source to sink in \textit{SN}
		have equal length, $h$; this means that the set of all nodes of \textit{SN} can
		be partitioned into $h+1$ layers $0, \ldots, h$, where layer $0$ is the source layer
		and layer $h$ the sink layer.
		Denote by $\textit{SN}'$ the social network obtained from \textit{SN}
		by removing the sink layer;
		we now define the labelling $\auxLabelling(s)$ of $\textit{SN}'$ as follows.
		Notice that every dual pair is contained in one layer.
		Consider first any labelling of $\textit{SN}'$ where the $n$ dual pairs in layer $0$ are assigned
		the values that represent $s$.
		The network \textit{SN} converges after $h-1$ updates by Proposition~\ref{acyclic}.
		We then pick as $\auxLabelling(s)$ the labelling of the limit network.

		\paragraph*{Construction of network \MainNet and labelling \mainLabelling.}
		
		We construct a social network from the components described above.
		Given a Turing machine $M$, we take the network \textit{SN} from Lemma~\ref{l:step}
		and combine it with the fuse line, valve, and alarm as follows:
		\begin{itemize}
			
			\item
			For each $i = 1, \ldots, n$,
			the $i^{th}$ source dual pair of \textit{SN} is identified with
			the $i^{th}$ sink   dual pair of \textit{SN}.
			(This transforms \textit{SN} into a cyclic network, where all cycles have length
			divisible by~$h$.)
			
			\item
			Every dual pair in \textit{SN} connects to a distinct pair in the fuse line.
			(As described above.)
			
			\item
			Every node in the alarm has edges to all dual pairs in \textit{SN},
			except for the fuse line and pair $Q$ of the valve (in other words, to all dual pairs and to pair $P$).
			(As described above.)
			
		\end{itemize}
		The fuse line needs as many pairs as there are
		dual pairs in \textit{SN},
		and $k$ can be chosen as $2$ (based on the proof of Lemma~\ref{l:alarm}).
		This completes the construction of the network \MainNet in our reductions.
		
		Given a configuration $s \in \{0, 1\}^n$ of the Turing machine $M$,
		consider any labelling of \MainNet that satisfies the following conditions:
		(\textit{i})   nodes in \textit{SN} are labelled according to the auxiliary labelling $\auxLabelling(s)$ defined above;
		(\textit{ii})  each pair in the fuse line and the valve is valid (i.e., its nodes disagree);
		(\textit{iii}) the alarm is evenly split; and
		(\textit{iv}) in every connection of the form shown in Figure~\ref{fig:fuse-connect},
		exactly $2$ out of $4$ intermediate nodes have value~$0$.
		Denote this labelling by \mainLabelling.
		
		\subsection{Hardness proofs}
		Let us proceed to proving the computational hardness of the problems for  
		Theorems~\ref{th:fixed-term} and~\ref{th:quant-term}.
		\paragraph*{Proof of Theorem~\ref{th:fixed-term}.}
		
		We already argued membership in PSPACE above and will prove hardness here.
		We rely on the fact that
		%
		there exists a {universal, polynomial-space bound} Turing machine $U$ for which the following problem is PSPACE-complete:
		\begin{quote}
			\emph{Input:}
			an integer $n \ge 1$ and a configuration $s^{(0)} \in \{0, 1\}^n$ of $U$.\\
			\emph{Output:}
			does $U$ diverge when started from configuration $s^{(0)}$?
		\end{quote}
		The complexity of this problem is shown similarly to Exercise~4.1 in~\cite{AroraBarak}.
		See also Theorem 19.9 in~\cite{PapadimitriouCC}.
		
		Apply the construction above to the Turing machine $U$. 
		Take the network \MainNet and the labelling $\mainLabelling$ defined above.
		First note that dual pairs in \textit{SN} have inputs from inside \textit{SN}
		and $2 k$ inputs from the alarm. This means that \textit{SN} will function
		``autonomously'' as long as the alarm remains evenly split.
		By Lemma~\ref{l:step}, \textit{SN} will in this case compute consecutive configurations
		of the Turing machine $U$.
		There is a ``pipelining'' effect involved: the labelling of the source
		level of \textit{SN} will be set to $s^{(0)}$ at time~$0$, then to $s^{(1)}$, the
		successor of $s^{(0)}$, at times~$1, \ldots, h$, then to $s^{(2)}$
		for the next $h$ steps, then to $s^{(3)}$, etc.
		
		Observe that
		if the Turing machine $U$ diverges when started from the configuration $s^{(0)}$,
		then, by the above, the alarm will always remain evenly split,
		flipping forever. This means that
		\MainNet does not converge.
		On the other hand, if $U$ terminates, than some dual pair will become invalid
		(Lemma~\ref{l:step}), the alarm will go off (Lemmas~\ref{l:stuck} and~\ref{l:alarm}),
		and the network will converge.
		Theorem~\ref{th:fixed-term} follows.
		
		\paragraph*{Proof of Theorem~\ref{th:quant-term}.}
		
		Again, we already argued membership in PSPACE above and will prove hardness here.
		We will now rely on PSPACE-completeness of
		the following problem: 
		\begin{quote}
			\emph{Input:}
			an integer $n \ge 1$ and (a description of) a Turing machine $M$.\\
			\emph{Output:}
			is there a configuration $s \in \{0, 1\}^n$ such that
			$M$ diverges when started from $s$?
		\end{quote}
		The hardness of this problem is a straightforward variation of the Corollary of Theorem~19.9 in~\cite{PapadimitriouCC}.
		
		The proof of Theorem~\ref{th:quant-term} extends the proof of Theorem~\ref{th:fixed-term}.
		Instead of $U$, we now have a Turing machine~$M$.
		Recall from the previous proof that
		if there is a configuration $s \in \{0, 1\}^n$ from which $M$ diverges,
		then there is an initial labelling from which \MainNet fails to converge.
		So we will now consider
		the case where $M$ terminates started from every configuration.
		Can there now be a labelling from which \MainNet fails to converge?
		
		To answer this question,
		let us look into various initial labellings of \MainNet.
		Let $\someLabelling$ such a labelling.
		By Lemma~\ref{l:alarm},
		if there exists a time $t \in \{0, 1, \ldots, h-1\}$ for which
		the network \MainNet has an invalid dual pair,
		then \MainNet converges.
		The same holds if \MainNet has an invalid pair
		in the fuse line or valve, or if the alarm is not evenly split.
		
		Suppose none of the above applies;
		then consider configurations $s_0, \ldots, s_{h-1} \in \{0,1\}^n$
		formed by the values of the source-layer dual pairs of \textit{SN}
		at times $0, 1, \ldots, h-1$.
		By the arguments above, the network \MainNet simulates the Turing machine~$M$
		in the following way.
		For each $i \in \{0, 1, \ldots, h-1\}$,
		at times $t \in \{ i, i + h, i + 2 h, \ldots \}$ the source-layer dual pairs of \textit{SN}
		form consecutive configurations of $M$ started from $s_i$.
		If $M$ terminates when started from some $s' \in \{s_0, \ldots, s_{h-1}\}$,
		then \MainNet converges when started from the labelling $\someLabelling$.
		This means that a necessary condition for \MainNet to fail to converge
		(starting from $\someLabelling$) is
		that
		$M$ diverges when started from every $s_i$, $i \in \{0, 1, \ldots, h-1\}$.
		In this case,
		there certainly \emph{exists}
		a configuration $s_i$ from which the Turing machine~$M$
		diverges.
		This completes the proof of Theorem~\ref{th:quant-term}.

		\section{Conclusions}

		We have shown that checking convergence of opinion diffusion in social networks is PSPACE-complete.
		Our results extend to majority-based multi-issue opinion diffusion \cite{GrandiLP15}, also in presence of integrity constraints \cite{BotanEtAlCOMSOC2018}, and to all {update rules that admit suitable modification of our gadgets, such as quota rules (in which an agent switches an opinion if a specified fraction of their influencers disagrees with them).}
		
		
		There are many possible directions for further research. 
		First, we have noted how some classes of networks, e.g., DAGs, are convergent and this can be verified efficiently. Our results imply that there is no efficiently computable characterisation of convergent networks, however we can ask whether a meaningful characterisation exists for networks that converge fast. 
		Second, an interesting question is whether the existence of a {\em non-trivial} (i.e., different from all-0 and all-1) fixed-point configuration in our model is an NP-complete property. 		
		Third, we have limited ourselves to the study of synchronous opinion diffusion protocols. This is possibly the simplest social network update model, widely adopted in the literature. It is also of interest what happens in asynchronous networks. We note that losing synchronicity makes the system nondeterministic, so the question of convergence changes significantly. We would for example need to study different forms of convergence, e.g., for all possible update orderings, for some, and the like.
		%
		Finally, our results are based on a worst-case complexity analysis and an important question remains regarding the complexity of verifying convergence in random networks.  
		
		\bibliography{literature1}
	\end{document}